\documentclass[pra,aps,twocolumn,superscriptaddress,a4paper]{revtex4-1}
\usepackage{amsmath,amsthm,amssymb,amsfonts,mathrsfs,amsthm,centernot}
\usepackage{dsfont}
\usepackage{graphicx,hyperref}
\usepackage{enumitem}
\usepackage[perpage,symbol*]{footmisc}
\usepackage{color}
\usepackage{soul}

\newtheorem*{theorem*}{Theorem}

\def\Tr{\operatorname{Tr}}
\def\>{\rangle}
\def\<{\langle}
\def\N#1{\left|\!\left|{#1}\right|\!\right|}

\def\mE{\mathcal{E}}

\def\openone{\mathds{1}}
\newcommand{\set}[1]{\mathcal{#1}}
\newcommand{\op}[1]{\mathsf{#1}}

\newcommand{\defeq}{\triangleq}

\newcommand{\cT}{\Gamma}

\newcommand{\psucc}{P_{\operatorname{succ}}}

\newcommand{\supa}{\operatorname{sup}_{\alpha}}
\newcommand{\infa}{\operatorname{inf}_{\alpha}}

\renewcommand{\qedsymbol}{\nobreak \ifvmode \relax \else
	\ifdim \lastskip<1.5em \hskip-\lastskip \hskip1.5em plus0em
	minus0.5em \fi \nobreak \vrule height0.75em width0.5em
	depth0.25em\fi}

\renewcommand{\ge}{\geqslant}
\renewcommand{\le}{\leqslant}
\renewcommand{\geq}{\geqslant}
\renewcommand{\leq}{\leqslant}

\newtheorem{theorem}{Theorem}
\newtheorem{corollary}{Corollary}
\newtheorem{lemma}{Lemma}

\newtheorem{definition}{Definition}

\newtheorem*{corollary*}{Corollary}
\newtheorem*{lemma*}{Lemma}
\newtheorem*{proposition*}{Proposition}
\newtheorem*{postulate*}{Postulate}
\newtheorem*{definition*}{Definition}

\theoremstyle{remark}
\newtheorem{remark}{Remark}
\newtheorem*{remark*}{Remark}

\theoremstyle{definition}

\newcommand{\bea}{\begin{eqnarray}}
\newcommand{\eea}{\end{eqnarray}}
\newcommand{\be}{\begin{equation}}
\newcommand{\ee}{\end{equation}}

\def\be{\begin{equation}}
\def\ee{\end{equation}}
\newcommand{\ba}{\begin{equation}\begin{aligned}}
\newcommand{\ea}{\end{aligned}\end{equation}}

\newcommand{\mH}{\mathcal{H}}

\newcommand{\la}{\langle}
\newcommand{\ra}{\rangle}
\newcommand{\tr}{{\rm Tr}}

\newcommand{\mbb}[1]{\mathbb{#1}}

\begin{document}

\title{Quantum Relative Lorenz Curves}
\author{Francesco \surname{Buscemi}}
\affiliation{Department of Computer Science and Mathematical Informatics, Nagoya University, Chikusa-ku, Nagoya,
	464-8601, Japan}
\email{buscemi@is.nagoya-u.ac.jp}
\author{Gilad \surname{Gour}}
\affiliation{Institute for Quantum Science and Technology and Department of Mathematics and Statistics,
	University of Calgary, 2500 University Drive NW, Calgary, Alberta, Canada T2N 1N4}
\email{gour@ucalgary.ca}

\begin{abstract}
The theory of majorization and its variants, including thermomajorization, have been found to play a central role in the formulation of many physical resource theories, ranging from entanglement theory to quantum thermodynamics. Here we formulate the framework of quantum relative Lorenz curves, and show how it is able to unify majorization, thermomajorization, and their noncommutative analogues. In doing so, we define the family of Hilbert $\alpha$-divergences and show how it relates with other divergences used in quantum information theory. We then apply these tools to the problem of deciding the existence of a suitable transformation from an initial pair of quantum states to a final one, focusing in particular on applications to the resource theory of athermality, a precursor of quantum thermodynamics.
\end{abstract}

\maketitle


\section{Introduction}

Lorenz curves, originally introduced to give a quantitative and pictorially clear representation of the inequality of the wealth distribution in a country~\cite{lorenz_methods_1905}, have since then been used also in other contexts in order to effectively compare different distributions (see for example~\cite{Marshall2011} and references therein). In their typical formulation, Lorenz curves fully capture the notion of ``nonuniformity''~\cite{Gour20151} of a distribution, in the sense that comparing the Lorenz curves associated to two given distributions (say, $p$ and $q$) induces an ordering equivalent to the relation of \textit{majorization}, which, in turns, is well-known to be equivalent to the existence of a random permutation (i.e., a bistochastic channel) transforming $p$ into $q$~\cite{Marshall2011}.

More recently, some variants of the original definition were proposed in order to capture other aspects of a given distribution, besides its mixedness. In particular, \textit{thermomajorization} was introduced in~\cite{horodecki_fundamental_2013} to characterize state transitions under thermal operations or Gibbs preserving operations~\cite{brandao_resource_2013}. Here, the corresponding Lorenz curve characterizes a partial ordering relative to the Gibbs distribution, rather than the uniform one.

This suggests that Lorenz curves are best understood not as properties of one given distribution, but rather of a given \textit{pair} of distributions, one being the ``state'' at hand and the other being the ``reference''. For example, the original Lorenz curve contains information about a given distribution $p$ with respect to the uniform one: it is in this precise sense, then, that the Lorenz curve characterizes the degree of nonuniformity of $p$---exactly because the reference distribution is chosen to be the uniform one. In the same way, thermomajorization measures the  degree of ``athermality'' because, in this case, the reference distribution is chosen to be the thermal (Gibbs) distribution.

A lot of attention has been devoted recently to the generalization of the above ideas to the case in which, rather than comparing  distributions, one wants to compare \textit{quantum states}, namely, density operators defined on a Hilbert space. This is one of the topics lying at the core of theories like quantum thermodynamics and, more generally, quantum resource theories~\cite{Hor13,Spe16,Bra15}.
However, a general theory of quantum Lorenz curves would be interesting in its own right, providing new insights on the rich analogies existing between quantum theory and classical probability theory, despite their differences. 

In this paper we develop such a theory by introducing the notions of quantum testing region, quantum relative Lorenz curves, and quantum relative majorization in much analogy with their classical counterparts. We find equivalent conditions for quantum relative majorization in terms of a new family of divergences that we call Hilbert $\alpha$-divergences, with $\alpha\in(1,\infty)$, and show that in the limits $\alpha\to1$ and $\alpha\to\infty$ the Hilbert $\alpha$-divergences are equivalent to the trace-distance and the max-relative entropy, respectively. As an application to quantum thermodynamics, we show that only the min- and max-relative entropies are needed to determine whether it is possible to convert one qubit athermality resource to another by Gibbs preserving operations. Finally, we show that in higher dimensions, quantum relative Lorenz curves can be used to determine the existence of a test-and-prepare channel converting one pair of states to another. 

\section{Quantum relative Lorenz curves}

Consider the task of distinguishing which, among two possible distributions, is the one that originated a set of observed sample data. This scenario, central in statistics, is usually treated within the framework known as \textit{hypothesis testing}~\cite{neyman_problem_1933}: the two distributions are called the \textit{null hypothesis} and the \textit{alternative hypothesis}, respectively, and the task of the statistician is to minimize the so-called \textit{type II error} (i.e., the probability of wrongly accepting the null hypothesis, namely, the probability of false negatives) given that the \textit{type I error} (i.e., the probability of wrongly rejecting the null hypothesis, namely, the probability of false positives) falls below a certain threshold. The whole hypothesis testing problem is hence ``encoded'' in the shape of the region of the $xy$-plane containing all achievable points $(x,y)=(\text{type I},\text{type II})$. Such a region is, by construction, convex, always contains the points $(0,0)$ and $(1,1)$, and is symmetric, in the sense that $(x,y)$ belongs to the region if and only if $(1-x,1-y)$ does, as this corresponds to exchanging the roles of null and alternative hypotheses (see Fig.~\ref{fig:region-class}). Hence, the hypothesis testing is fully characterized by the upper boundary of the region. In particular, as noticed by Renes~\cite{renes_relative_2015}, when testing $p$ against the uniform distribution, such boundary coincides with the usual Lorenz curve; when testing $p$ against the Gibbs distribution, it coincides with the thermomajorization curve.

The observations in~\cite{renes_relative_2015} exhibit a fundamental connection between the theory of (thermo)majorization and hypothesis testing. It is then extremely natural for us here to introduce the definition of Lorenz curves for pair of quantum states, leveraging on the fact that hypothesis testing is well understood in the quantum case too~\cite{Aud07,Nus09,Hay07,Hia91,Oga00}:
\begin{definition}
	Given two density matrices $\rho_1$ and $\rho_2$ on $\mathbb{C}^n$, the associated \emph{testing region} $\set{T}(\rho_1,\rho_2)\subset\mathbb{R}^2$ is defined as the set of achievable points
	\begin{equation*}
		(x,y)=(\Tr[E\rho_2],\Tr[E\rho_1])\;,
	\end{equation*}
	with $0\le E\le\openone_n$. The \emph{quantum Lorenz curve} of $\rho_1$ relative to $\rho_2$ is defined as the upper boundary of $\set{T}(\rho_1,\rho_2)$, see Fig.~\ref{fig:region-quantum}.
\end{definition}

\begin{figure}[tp]
\includegraphics[width=7cm]{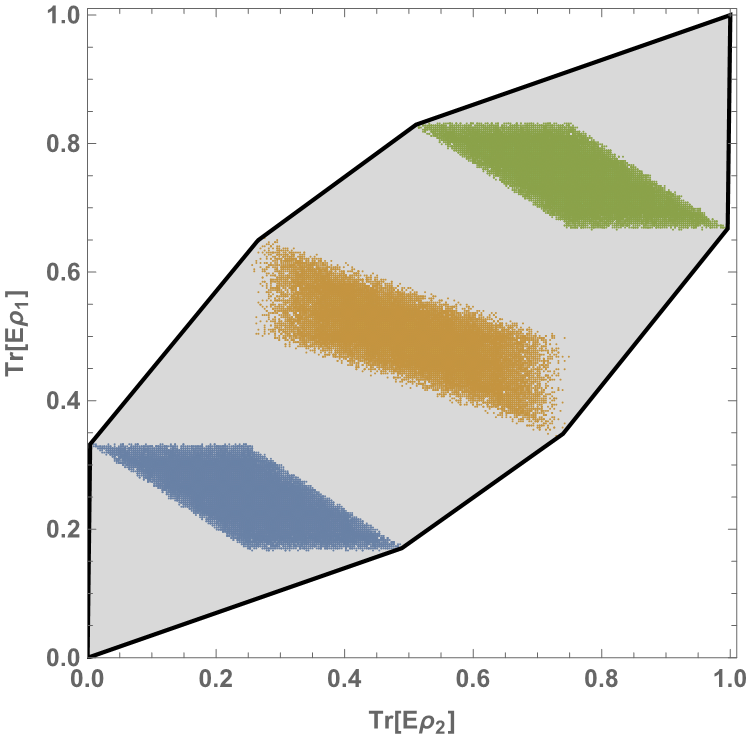}
\caption{Example of classical testing region in dimension $n=4$ with $\vec{p}_1=(1/2, 1/4, 1/4, 0)^T$ and $\vec{p}_2=(1/6, 1/6, 1/3, 1/3)^T$. The Lorenz curve (i.e. the upper boundary) is determined by the vertices at which the Lorenz curve changes slope.}
\label{fig:region-class}
\end{figure}

Closely related to the testing region, the \textit{hypothesis testing relative entropy} (see e.g.~\cite{Wang12,Dup13,Datta14} and references therein) is defined, for $0\le\epsilon\le1$, as follows:
\begin{align}
D^\epsilon_{H}(\rho_1\| \rho_2) &\defeq-\log Q^\epsilon(\rho_1\|\rho_2)\nonumber\\
Q^\epsilon(\rho_1\|\rho_2) &\defeq\min_{\substack{0\leq E\leq\openone_n\\ \tr[\rho_1 E]\geq 1-\epsilon }}\tr[\rho_2 E].\label{eq:Qeps}
\end{align}
As noted in Ref.~\cite{Dup13}, the computation of $Q^\epsilon(\rho_1\|\rho_2)$ can be solved efficiently by semidefinite linear programming (SDP). In fact, in what follows (see Eqs.~(\ref{eq:spoiler1}) and~(\ref{eq:spoiler2} in Section~\ref{sec:reform}) we show that, using the strong duality relation of SDP, it is possible to write $Q^\epsilon(\rho_1\|\rho_2)$, for any fixed $\epsilon$, as the maximum of a simple function of one real variable, namely, $Q^\epsilon(\rho_1\|\rho_2)=\max_{r\ge0}f_\epsilon(r)$, where
\ba\label{eq:how-to-qlc}
f_\epsilon(r) & \defeq(1-\epsilon)r-\tr(r\rho_1-\rho_2)_+\\
& =\frac{1}{2}\Big[1+(1-2\epsilon)r-\N{r\rho_1-\rho_2}_1\Big]\;.
\ea
This observation will play an important role in what follows, by considerably simplifying our analysis.
	
The above definition of relative Lorenz curve generalizes the classical Lorenz curve to the quantum case. In particular, if $\rho_1$ and $\rho_2$ commute, they can be simultaneously diagonalized, and the testing region in this case becomes the collection of points 
\[\set{T}_{\rm cl}(\vec{p_1},\vec{p_2})\defeq\{(\vec{t}\cdot\vec{p}_1,\vec{t}\cdot\vec{p}_2)\;:\;\vec{t}\in\mathbb{R}^{n}_{+}\;,\vec{t}\leq (1,1,...,1)^T\},\]
where $\vec{p}_1$ and $\vec{p}_2$ are the diagonals of $\rho_1$ and $\rho_2$ written in a vector form. In this case, Blackwell proved a very strong relation~\cite{blackwell_equivalent_1953,torgersen_comparison_1991}: given two pairs of distributions $(\vec{p}_1,\vec{p}_2)$ and $(\vec{q}_{1},\vec{q}_{2})$, the inclusion $
\set{T}_{\rm cl}(\vec{q}_{1},\vec{q}_{2})\subseteq\set{T}_{\rm cl}(\vec{p}_1,\vec{p}_2)$ holds if and only if there exists a column stochastic matrix $M$ such that $\vec{q}_{1}=M\vec{p}_1$ and $\vec{q}_{2}=M\vec{p}_2$. Known results about classical (thermo) majorization are therefore special cases of Blackwell's theorem, even though Blackwell's work actually predates some of them (see the discussion in Refs.~\cite{buscemi-gibbs,renes_relative_2015}).

In the rest of the paper we explore the extent to which statements similar to Blackwell's theorem can be proved in the quantum case. However, our interest here does not lie as much in the general case, for which we know that many classical results cease to hold~\cite{shmaya_comparison_2005,buscemi_comparison_2012,matsumoto_quantum_2010,jencova_comparison_2012,matsumoto_example_2014,buscemi-gibbs}, but rather in restricted scenarios of practical relevance, especially for the growing field of quantum resource theories.

\begin{figure}[tp]
	\includegraphics[width=7cm]{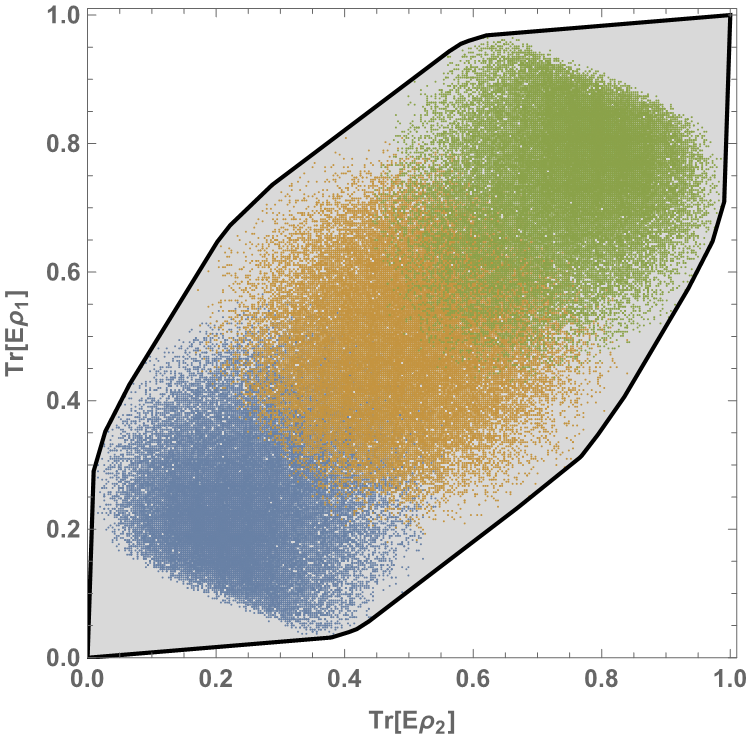}
	\caption{Numerical example of the quantum testing region for two random four-dimensional density matrices. Notice that the curve is only roughly approximated as the sampled measurements are not enough to determine it neatly. Quantum Lorenz curves are efficiently obtained by semi-definite linear programming, e.g., using Eq.~(\ref{eq:how-to-qlc}) in the main text.}
	\label{fig:region-quantum}
\end{figure}

\section{Hilbert $\alpha$-divergences}

In analogy with the notation used for majorization, we write
\begin{equation*}
	(\rho_1,\rho_2)\succ(\rho_1',\rho_2')
\end{equation*}
and say that $(\rho_1,\rho_2)$ \textit{relatively majorizes} $(\rho'_1,\rho'_2)$, whenever the quantum Lorenz curve of $\rho_1$ relative to $\rho_2$ lies everywhere above the quantum Lorenz curve of $\rho_1'$ relative to $\rho_2'$, that is, $\set{T}(\rho_1,\rho_2)\supseteq\set{T}(\rho'_1,\rho'_2)$.

As a tool to characterize quantum relative majorization, we introduce here a family of divergences as follows: given two density matrices $\rho$ and $\sigma$ on $\mathbb{C}^n$, we define, for all $\alpha\ge1$, the following quantity:
\begin{equation}\label{eq:sup-alpha}
	\operatorname{sup}_\alpha(\rho/\sigma)\defeq\sup_{\alpha^{-1}\openone_n\le E\le\openone_n}\frac{\Tr[E\rho]}{\Tr[E\sigma]},
\end{equation}
and the corresponding divergence:
\begin{equation}\label{eq:alpha-hil-div}
	H_\alpha(\rho\|\sigma)\defeq\frac{\alpha}{\alpha-1}\log_2\operatorname{sup}_\alpha(\rho/\sigma).
\end{equation}
The notation used in Eq.~(\ref{eq:sup-alpha}) is adapted from Refs.~\cite{bushell_hilberts_1973,eveson_hilberts_1995,reeb_hilberts_2011}: there the quantity
\begin{equation*}
	\begin{split}
		\sup(\rho/\sigma)&\defeq\inf\{\lambda:\lambda\sigma-\rho\ge0\}\\
		&=\lim_{\alpha\to\infty}\operatorname{sup}_\alpha(\rho/\sigma)
	\end{split}
\end{equation*}
is used to define the Hilbert projective metric
\begin{equation*}
	\mathfrak{h}(\rho,\sigma)\defeq\ln[\sup(\rho/\sigma)\sup(\sigma/\rho)]. 
\end{equation*}
We note that, in Ref.~\cite{reeb_hilberts_2011}, the quantity $\inf(\rho/\sigma)$ is also introduced, as $\sup\{\lambda:\rho-\lambda\sigma\ge 0 \} $: in our notation it coincides with $\inf_{0\le E\le\openone_n}\{\Tr[E\rho]/\Tr[E\sigma] \}=1/\sup(\sigma/\rho)$. Due to the relation with the Hilbert's metric, we refer to the divergences in Eq.~(\ref{eq:alpha-hil-div}) as \textit{Hilbert $\alpha$-divergences}. Their main properties are summarized in the following theorem:
\begin{theorem}\label{prop:1}
Let $\rho$ and $\sigma$ be two density matrices on $\mathbb{C}^n$. Then:
\begin{enumerate}[label=\roman*)]
	\item for all $\alpha\ge 1$, 
	$H_{\alpha}(\rho\|\sigma)\geq 0$,
	with equality if and only if $\rho=\sigma$;
	\item for all $\alpha\ge 1$, the data-processing inequality holds: for any (not necessarily completely) positive trace-preserving map $\Phi$,
	$H_{\alpha}(\Phi(\rho)\|\Phi(\sigma))\leq H_{\alpha}(\rho\|\sigma)$;
	\item $H_{\infty}(\rho\|\sigma)\defeq\lim_{\alpha\to\infty}H_\alpha(\rho\|\sigma)= D_{\max}(\rho\|\sigma)$, namely, the max-relative entropy of Ref.~\cite{datta_min-_2009};
	\item 
	$H_1(\rho\|\sigma)\defeq\lim_{\alpha\to 1}H_{\alpha}(\rho\|\sigma)=\frac{1}{2\ln(2)}\N{\rho-\sigma}_1$. 
	\end{enumerate}
\end{theorem}

\begin{remark}
$H_\alpha$ is thus a family of divergences connecting the trace-distance (when $\alpha\to1$) with $D_{\max}$ (when $\alpha\to\infty$). In passing by, we also notice that, while in point~(ii) above the data-processing inequality is stated to hold for any positive trace-preserving map, Hilbert $\alpha$-divergences are in fact monotonically decreasing for an even larger set of transformations, called \textit{2-statistical morphisms}: while this point is outside the scope of the present work, we refer the interested reader to Refs.~\cite{buscemi_comparison_2012,matsumoto_quantum_2010,jencova_comparison_2012,buscemi_less-noisy}.
\end{remark}

\begin{proof}
	Properties (ii) and (iii) are direct consequences of the definition of $\operatorname{sup}_\alpha(\rho/\sigma)$.
	
	In order to prove property (iv), we start by taking $\alpha>1$ and defining the following two quantities: $\epsilon\defeq\alpha-1$ and $\delta_\epsilon\defeq\operatorname{sup}_{1+\epsilon}(\rho/\sigma)-1$. With these notations, from the definition of $\supa(\rho/\sigma)$ we obtain
	\be\label{fr}
	\delta_\epsilon\geq\frac{\Tr[E\rho]}{\Tr[E\sigma]}-1= \frac{\Tr\left[E(\rho-\sigma)\right]}{\Tr\left[E\sigma\right]}\;,
	\ee
	for all $\frac{1}{1+\epsilon}\openone\leq E\leq \openone$.
	Introducing the operator
	\be
	\Delta\defeq\frac{1}{\epsilon}\left[(1+\epsilon)E-\openone\right]\;,
	\ee
	we get that Eq.~\eqref{fr} is equivalent to 
	\be
	\delta_\epsilon\geq \frac{\epsilon\Tr\left[\Delta(\rho-\sigma)\right]}{1+\epsilon\Tr\left[\Delta\sigma\right]}\;,
	\ee
	for all $0\leq \Delta\leq \openone$. Hence, $\lim_{\epsilon\to 0}\delta_\epsilon=0$.
	We therefore have
	\begin{align}
	\lim_{\epsilon\to 0}H_{1+\epsilon}(\rho\|\sigma)&=\frac{1}{\ln(2)}\lim_{\epsilon\to 0}\frac{1+\epsilon}{\epsilon}\ln(1+\delta_\epsilon)\nonumber\\
	&=\frac{1}{\ln(2)}\lim_{\epsilon\to 0}\frac{1}{\epsilon}\delta_\epsilon\nonumber\\
	&\geq \frac{1}{\ln(2)}\lim_{\epsilon\to 0}\frac{1}{\epsilon}\frac{\epsilon\Tr\left[\Delta(\rho-\sigma)\right]}{1+\epsilon\Tr\left[\Delta\sigma\right]}\nonumber\\
	&=\frac{1}{\ln(2)}\Tr\left[\Delta(\rho-\sigma)\right]\;,
	\end{align}
	for all $0\leq \Delta\leq \openone$. 
	We therefore conclude that 
	\be
	H_{1}(\rho\|\sigma)\geq \frac{1}{\ln(2)}\Tr[(\rho-\sigma)_+]=\frac{1}{2\ln(2)}\N{\rho-\sigma}_1
	\ee
	where we chose $\Delta$ to be the projection to the positive part of $\rho-\sigma$. To see that $H_{1}(\rho\|\sigma)=\frac{1}{2\ln 2}\N{\rho-\sigma}_1$ note that, in fact, by definition
	\begin{equation*}
	\begin{split}
		\delta_\epsilon&=\max_{\frac{1}{1+\epsilon}\openone\leq E\leq \openone} \frac{\Tr\left[E(\rho-\sigma)\right]}{\Tr\left[E\sigma\right]}\\
		&=\max_{0\leq \Delta\leq \openone} \frac{\epsilon\Tr\left[\Delta(\rho-\sigma)\right]}{1+\epsilon\Tr\left[\Delta\sigma\right]}\\
		&=\epsilon\max_{0\leq \Delta\leq \openone} \Tr\left[\Delta(\rho-\sigma)\right]+O(\epsilon^2)
	\end{split}
	\end{equation*} 
	Hence, in the limit $\epsilon\to 0$ we get $\lim_{\epsilon\to 0}\frac{1}{\epsilon}\delta_{\epsilon}=\Tr[(\rho-\sigma)_+]=\frac12\N{\rho-\sigma}_1$.
	
	Finally, property (i) is proved as follows. Since $\sup_{\alpha}(\rho\|\sigma)\geq 1$ ,
	we always have $H_\alpha(\rho\|\sigma)\geq 0$. 
	For $\alpha>1$ if $H_\alpha(\rho\|\sigma)=0$ then $\supa(\rho\|\sigma)=1$. Hence, $\Tr[E\rho]\leq\Tr[E\sigma]$ for all $\alpha^{-1} \openone\leq E\leq \openone$. Introducing
	\begin{equation*}
		\Delta\defeq\frac{\alpha}{\alpha-1}(E-\frac{1}{\alpha}\openone),
	\end{equation*}
	we get $\Tr[\Delta\rho]\leq\Tr[\Delta\sigma]$, namely, $\Tr[\Delta\;(\rho-\sigma)]\le 0$, for all $0\leq \Delta\leq \openone$. We therefore must have $\rho=\sigma$. The case $\alpha=1$ follows from property (iv).
\end{proof}

\section{Relative majorization as sets of inequalities}\label{sec:reform}

We are now in a position to provide a set of alternative conditions, reformulating the relative majorization ordering $(\rho_1,\rho_2)\succ(\rho'_1,\rho'_2)$ as sets of inequalities.
\begin{theorem}\label{theo:1}
	Consider two pairs of density matrices $(\rho_1,\rho_2)$ on $\mathbb{C}^n$ and $(\rho'_1,\rho'_2)$ on $\mathbb{C}^m$. The following are equivalent:
	\begin{enumerate}[label=\roman*)]
		\item $(\rho_1,\rho_2)\succ(\rho'_1,\rho'_2)$;
		\item for all $t\ge0$, $\N{\rho_1-t\rho_2}_1\ge\N{\rho'_1-t\rho'_2}_1$;
		\item for all $\alpha\ge1$,
		\[
		\begin{cases}
		H_\alpha(\rho_1\|\rho_2)\ge H_\alpha(\rho'_1\|\rho'_2),\\  H_\alpha(\rho_2\|\rho_1)\ge H_\alpha(\rho'_2\|\rho'_1);
		\end{cases}
		\]
		\item for all $0\le\epsilon\le1$, $D^\epsilon_{H}(\rho_1\| \rho_2)\geq D^\epsilon_{H}(\rho_1'\| \rho_2')$.
	\end{enumerate}
\end{theorem}

We split the proof into several lemmas.

\begin{lemma}\label{lem:separation}
	Given two pairs of density operators $(\rho_1,\rho_2)$ and
	$(\rho_1',\rho_2')$ on $\mathbb{C}^n$ and $\mathbb{C}^m$, respectively, the following are equivalent:
	\begin{enumerate}[label=\roman*)]
		\item $(\rho_1,\rho_2)\succ(\rho'_1,\rho'_2)$;
		\item $\set{T}(\rho_1,\rho_2)\supseteq\set{T}(\rho'_1,\rho'_2)$;
		\item $\N{t_1\rho_1+t_2\rho_2}_1\ge
		\N{t_1\rho_1'+t_2\rho_2'}_1$ for all $t_1,t_2\in\mathbb{R}$;
		\item $\N{\rho_1-t\rho_2}_1\ge \N{\rho'_1-t\rho_2'}_1$, for
		all $t\ge 0$.
	\end{enumerate}
\end{lemma}
\begin{proof}
	The first equivalence holds by definition. Denoting by $(p,\bar p)$ and $(q,\bar q)$ the generic element of $\set{T}(\rho_1,\rho_2)$ and $\set{T}(\rho_1',\rho_2')$, respectively,
	the Separation Theorem for convex sets, applied to
	$\set{T}(\rho_1,\rho_2)$ and $\set{T}(\rho_1',\rho_2')$, states that
	$\set{T}(\rho_1,\rho_2)\supseteq \set{T}(\rho_1',\rho_2')$ if and only if, for
	any $v=(a,b)\in\mathbb{R}^2$,
	\begin{equation}\label{eq:cond-sep-th}
	\max_{(p,\bar p)\in \set{T}(\rho_1,\rho_2)}\left[a p +b\bar p\right]\ge \max_{(q,\bar q)\in \set{T}(\rho_1',\rho_2')}\left[a q +b\bar q\right].
	\end{equation}
	The next step is to show that
	\begin{equation}
	\max_{(p,\bar p)\in \set{T}(\rho_1,\rho_2)}\left[a p +b\bar p\right]=\frac{a+b+\N{a\rho_1-b\rho_2}_1}2,
	\end{equation}
	and, analogously, for $(\rho_1',\rho_2')$. This is done by the following
	simple passages:
	\begin{equation}
	\begin{split}
	\max_{(p,\bar p)\in \set{T}(\rho_1,\rho_2)}\left[a p +b\bar
	p\right]&=\max_{0\le E\le \openone}\left\{a\Tr[\rho_1\ E]+b\Tr[\rho_2\ E]\right\}\\
	&=\max_{0\le
		E\le\openone}\Tr[\left(a\rho_1+b\rho_2\right)\ E]\\
	&=\Tr\left(a\rho_1+b\rho_2\right)_+,
	\end{split}
	\end{equation}
	where the last expression denotes the positive part of the
	self-adjoint operator $a\rho_1+b\rho_2$. Then, since $2\Tr(A)_+=\N{A}_1+\Tr[A]$ for any self-adjoint operator, we have that
	\begin{equation}
	2\max_{(p,\bar p)\in \set{T}(\rho_1,\rho_2)}\left[a p +b\bar
	p\right]=a+b+\N{a\rho_1+b\rho_2}_1.
	\end{equation}
	This proves that Eq.~(\ref{eq:cond-sep-th}) is satisfied if and
	only if $\N{a\rho_1+b\rho_2}_1\ge \N{a\rho_1'+b\rho_2'}_1$, for all $a,b\in\mathbb{R}$.

	We are left to prove that (iii) is equivalent to (iv). However, since (iv) is a special case of (iii), we only need to prove that (iv) implies (iii). To this end, we notice that, whenever $t_1,t_2\ge 0$ or $t_1,t_2\le 0$, $\N{t_1\rho_1+t_2\rho_2}_1=
	\N{t_1\rho'_1-t_2\rho_2'}_1$ always, simply due to the positivity
	of $\rho_1,\rho_2,\rho_1',\rho_2'$. We can hence consider only the
	cases $t_1> 0> t_2$ or $t_2>
	0> t_1$. However, since
	$\N{X}_1=\N{-X}_1$, for any matrix $X$, we can further
	restrict the parameters $t_1$ and $t_2$ to the case $t_2< 0< t_1$. The statement is finally obtained by rescaling both
	$t_1$ and $t_2$ by the (positive) factor $1/t_1$.
\end{proof}

Lemma~\ref{lem:separation} above shows that statements (i) and (ii) of Theorem~\ref{theo:1} are indeed equivalent. We now move on to proving the equivalence of the point (iii). We begin with the following lemma.

\begin{lemma}\label{p3}
	For any choice of density operators $\rho$ and $\sigma$,
	\begin{equation*}
	\supa(\rho/\sigma)=\inf\left\{\lambda\geq 1\;:\;\frac{\N{\lambda\sigma-\rho}_1}{\lambda-1}\leq\frac{\alpha+1}{\alpha-1} \right\}\;.
	\end{equation*}
\end{lemma}

\begin{proof}
	Note first that
	\begin{equation*}
	\begin{split}
		&\supa(\rho/\sigma)\\	
		&\defeq\sup_{\alpha^{-1}\openone\le E\le\openone}\left\{{\Tr[E\rho]}/{\Tr[E\sigma]}\right\} \\
		&=\inf\{\lambda:\lambda\ge{\Tr[E\rho]}/{\Tr[E\sigma]}\textrm{ for all }\alpha^{-1}\openone\le E\le\openone \}\\
		& =\inf\left\{\lambda:\Tr\left[E\left(\lambda\sigma-\rho\right)\right]\geq 0\textrm{ for all }\alpha^{-1}\openone\le E\le\openone \right\}\\
		&=\inf\left\{\lambda\in\mbb{R}\;:\;\alpha^{-1}\Tr\left[\left(\lambda\sigma-\rho\right)_+\right]\geq\Tr\left[\left(\lambda\sigma-\rho\right)_-\right] \right\}\;,		
	\end{split}	
	\end{equation*}
	where, in the last equality, we used the decomposition $A=A_+-A_-$ for Hermitian operators and the choice $E=\alpha^{-1}\Pi_++\Pi_-$, being $\Pi_\pm$ the projectors onto the positive and negative parts of $(\lambda\sigma-\rho)$, respectively. Indeed, this is choice for the operator $E$ that poses the toughest constraints compatible with the fixed value of the parameter $\alpha$. (Equivalently, if $\Tr[E(\lambda\sigma-\rho)]\ge0$ for such a choice of $E$, then it is positive for any $\alpha^{-1}\openone\le E\le \openone$.)
	
	Then, using the relations
	\begin{equation*}
	\lambda-1=\Tr\left[\left(\lambda\sigma-\rho\right)_+\right]-\Tr\left[\left(\lambda\sigma-\rho\right)_-\right]
	\end{equation*}
	and
	\begin{equation*}
	\N{\lambda\sigma-\rho}_1=\Tr[(\lambda\sigma-\rho)_+]+\Tr[(\lambda\sigma-\rho)_-]
	\end{equation*}
	gives
	\begin{equation*}
	\supa(\rho\|\sigma)=\inf\left\{\lambda\in\mbb{R}\;:\;\Tr\left[\left(\lambda\sigma-\rho\right)_-\right]\leq\frac{\lambda-1}{\alpha-1} \right\}.
	\end{equation*}
	Then, since $\Tr\left[\left(\lambda\sigma-\rho\right)_-\right]\leq\frac{\lambda-1}{\alpha-1}$ if and only if $\N{\lambda\sigma-\rho}_1\le\frac{\lambda-1}{\alpha-1}+\Tr[(\lambda\sigma-\rho)_+]=\frac{\lambda-1}{\alpha-1}+\frac{\lambda-1+\N{\lambda\sigma-\rho}_1}{2}$, after an easy manipulation we obtain
	\[
	\supa(\rho/\sigma)=\inf\left\{\lambda\in\mathbb{R}\;:\;\frac{\N{\lambda\sigma-\rho}_1}{\lambda-1}\leq\frac{\alpha+1}{\alpha-1} \right\}\;.
	\]
	The statement is finally recovered by noticing that no loss of generality comes from restricting $\lambda$ to values greater than or equal to 1.
\end{proof}

\begin{lemma}\label{3}
		For any choice of density operators $\rho$ and $\sigma$, the function
	\[
	f(\lambda)=\frac{\N{\lambda\sigma-\rho}_1}{\lambda-1}
	\]
	is monotonically non-increasing in the domain $\lambda\ge1$ with $f(1)=\infty$ and $f(\infty)=1$.
\end{lemma}

\begin{remark}\label{rem:supa}
	In particular, Lemma~\ref{p3} and Lemma~\ref{3} above imply that, for any pair of density operators $\rho$ and $\sigma$,
	\[
	\frac{\N{\supa(\rho/\sigma)\sigma-\rho}_1}{\supa(\rho/\sigma)-1}=\frac{\alpha+1}{\alpha-1}.
	\]
\end{remark}

\begin{proof}
	Set $t= 1/(\lambda-1)$ and define $g(t)=\N{(1+t)\sigma-t\rho}_1$. Hence, $g(t)=f(\lambda)$, and it is enough to show that $g(t)$ is monotonically non-decreasing in its domain $t\in[0,\infty)$. First note that for any $0<p<1$ and $t_1,t_2\in\mbb{R}_+$ we have
	\begin{align*}
	& g\left(pt_1+(1-p)t_2\right) \\ 
	=&\N{p(1+t_1)\sigma-pt_1\rho+(1-p)(1+t_2)\sigma-(1-p)t_2\rho}_1\\
	\leq& \N{p(1+t_1)\sigma-pt_1\rho}_1+\N{(1-p)(1+t_2)\sigma-(1-p)t_2\rho}_1\\
	= &\; pg(t_1)+(1-p)g(t_2)\;.
	\end{align*}
	Hence $g(t)$ is a convex function. Moreover, note that $g(0)=1\leq g(t)$ for all $t\geq 0$. These two properties of $g(t)$ together imply that it is monotonically non-decreasing in $t$.
\end{proof}

\begin{lemma}\label{lem:sup-and-t}
	Consider two pairs of states $(\rho,\sigma)$ and $(\rho',\sigma')$. Then, the following are equivalent:
	\begin{enumerate}[label=\roman*)]
		\item for all $\alpha\ge1$, $\supa(\rho/\sigma)\geq\supa(\rho'/\sigma')$ and $\infa(\rho/\sigma)\le\infa(\rho'/\sigma')$;
		\item $\N{t\sigma-\rho}_1\geq\N{t\sigma'-\rho'}_1$ for all $t\ge 0$;
	\end{enumerate}
\end{lemma}
\begin{proof}
	We only need to show that the condition $\supa(\rho/\sigma)\geq\supa(\rho'/\sigma')$ for all $\alpha\ge1$ is equivalent to $\N{t\sigma-\rho}_1\geq\N{t\sigma'-\rho'}_1$ for all $t\ge1$. Then, if this holds, the remaining statement, namely, that $\infa(\rho/\sigma)\le\infa(\rho'/\sigma')$ for all $\alpha\ge 1$ is equivalent to $\N{t\sigma-\rho}_1\geq\N{t\sigma'-\rho'}_1$ for all $t\in[0,1]$, simply follows from the definitions.
	
	Set $M_\alpha\equiv \supa(\rho/\sigma)$, $M_\alpha'\equiv\supa(\rho'/\sigma')$, and recall the definition of $f$ in Lemma~\ref{3}.
	Then, as noticed in Remark~\ref{rem:supa} above, it follows that
	\begin{equation}
	f(M_\alpha)=\frac{\N{M_\alpha\sigma-\rho}_1}{M_\alpha-1}=\frac{\alpha+1}{\alpha-1}=\frac{\N{M_\alpha'\sigma'-\rho'}_1}{M_\alpha'-1}.
	\end{equation}
	Since $f$ is monotonically non-increasing, we get that $M_\alpha\geq M_\alpha'$ implies that 
	\be
	\frac{\N{M_\alpha'\sigma'-\rho'}_1}{M_\alpha'-1}\geq \frac{\N{M_\alpha\sigma'-\rho'}_1}{M_\alpha-1}.
	\ee
	Combining the above two equations gives
	\be
	\N{M_\alpha\sigma-\rho}_1\geq \N{M_\alpha\sigma'-\rho'}_1\quad\forall\;\alpha\ge1.
	\ee
	
	We now make the simple observation that, by definition, the function $\sup_\alpha(\rho/\sigma)=M_\alpha$ is continuous and monotonically nondecreasing in $\alpha$, with $\sup_{\alpha=1}(\rho/\sigma)=1$ and $\sup_{\alpha\to\infty}(\rho/\sigma)=\sup(\rho/\sigma)$. Hence,
	\be
	\N{t\sigma-\rho}_1\geq \N{t\sigma'-\rho'}_1,\quad1\le\forall t\le\sup(\rho/\sigma)\;,
	\ee
	and the above is enough to conclude that the same ordering holds in fact for all $t\ge1$.
	
	Conversely, suppose the inequality above holds for all $t\geq 1$. This implies that,
	if 
	\be
	\frac{\N{\lambda\sigma-\rho}_1}{\lambda-1}\leq\frac{\alpha+1}{\alpha-1},
	\ee
	then also
	\be
	\frac{\N{\lambda\sigma'-\rho'}_1}{\lambda-1}\leq\frac{\alpha+1}{\alpha-1}\;.
	\ee
	But from Lemma~\ref{p3} this implies that $\supa(\rho'/\sigma')\le\supa(\rho/\sigma)$. 
\end{proof}
Lemma~\ref{lem:sup-and-t} above hence proves the equivalence of point~(ii) and point~(iii) of Theorem~\ref{theo:1}, because $\sup_\alpha(\rho_1/\rho_2)\ge\sup_\alpha(\rho'_1/\rho'_2)$ if and only if $H_\alpha(\rho_1\|\rho_2)\ge H_\alpha(\rho'_1\|\rho'_2)$, and $\inf_\alpha(\rho_1/\rho_2)\le\inf_\alpha(\rho'_1/\rho'_2)$ if and only if $H_\alpha(\rho_2\|\rho_1)\ge H_\alpha(\rho'_2\|\rho'_1)$.

The proof of Theorem~\ref{theo:1} is complete if we  prove the equivalence of the remaining point (iv). Also in this case, rather than proving the statement for $D^\epsilon_H$, we will prove it for the corresponding $Q^\epsilon$, related with $D^\epsilon_H$ as given in Eq.~(\ref{eq:Qeps}) of the main text. We recall the definition: given two density operators $\rho_1$ and $\rho_2$ on $\mathbb{C}^n$, for any $\epsilon\in[0,1]$,
\be\label{eq:recall-def}
Q^\epsilon(\rho_1\|\rho_2)=\min_{\substack{0\leq A\leq\openone_n\\ \tr[\rho_1 A]\geq 1-\epsilon }}\tr[\rho_2 A].
\ee
For later convenience, we introduce the following notation: $\op{H}_n$ to denote the set of $n$-by-$n$ Hermitian matrices on $\mathbb{C}^n$, and $\op{M}_{n,+}$ to denote the set of $n$-by-$n$ complex positive semi-definite matrices.

\begin{lemma}\label{lem:5}
Given two pairs of density operators $(\rho_1,\rho_2)$ and
$(\rho_1',\rho_2')$ on $\mathbb{C}^n$ and $\mathbb{C}^m$, respectively, the following are equivalent:
	\begin{enumerate}[label=\roman*)]
		\item $Q^\epsilon(\rho_1\|\rho_2)\leq Q^\epsilon(\rho'_1\|\rho'_2)$, for all $\epsilon\in[0,1]$;
		\item $\N{\rho_2-r\rho_1}_1\geq\N{\rho'_2-r\rho'_1}_1$, for all $r\ge 0$.
	\end{enumerate}
\end{lemma}

\begin{proof}	
	Consider the following setting of linear programming. Let $V_1$ and $V_2$ be two (inner product) vector spaces with two cones $K_1\subset V_1$ and $K_2\subset V_2$. Consider two vectors $v_1\in V_1$ and $v_2\in V_2$, and a linear map $\cT:V_1\to V_2$. Given a problem in its primal form:
	\be
	\max_{\substack{x\in K_1\\
			v_2-\cT(x)\in K_2}}\langle v_1,x\rangle_1\;,
	\ee
	the dual form involves the adjoint map $\cT^*:V_2\to V_1$:
	\be\label{eq:dual-form}
	\min_{\substack{y\in K_2\\
			\cT^*(y)-v_1\in K_1}}\langle v_2,y\rangle_2\;,
	\ee
	where $\cT^*$ is defined by the relation $\<y,\cT(x)\>=\<\cT^*(y),x\>$, for all $x\in K_1$ and all $y\in K_2$.
	
	In our case, denote
	\[V_1\defeq \mbb{R}\oplus \op{H}_n=\left\{(r,A)\;\Big|\;r\in\mbb{R}\;\;;\;\;A\in\op{H}_n\right\}\]
	with inner product\[\left\langle (r,A),(t,B)\right\rangle_1\defeq rt+\tr[AB]\;.\]
	Further, define $K_1=\mbb{R}_{+}\oplus\op{M}_{n,+}$ to be the positive cone in $V_1$.
	Similarly, set $V_2=\op{H}_n$ and $K_2=\op{M}_{n,+}$. The linear map $\cT:V_1\to V_2$ is given by:
	\be
	\cT(r,A)=r\rho_1-A\;;
	\ee
	hence, the corresponding dual map $\cT^{*}:V_2\to V_1$ is given by
	\be
	\cT^*(B)=(\tr[\rho_1\; B],-B)\;.
	\ee
	Finally, set $v_1=(1-\epsilon,-\openone_n)$ and $v_2=\rho_2$.
	
	Since, for these choices, $y\in K_2$ if and only if $y\geq 0$ and $\cT^*(y)-v_1\in K_1$ if and only if $\tr[\rho y]\geq 1-\epsilon$ and $y\leq\openone_n$, the dual form~(\ref{eq:dual-form}) becomes exactly the right-hand side of Eq.~(\ref{eq:recall-def}), namely,
	\be
	\min_{\substack{y\in K_2\\
			\cT^*(y)-v_1\in K_1}}\langle v_2,y\rangle_2=Q^\epsilon(\rho_1\|\rho_2)\;.
	\ee
	For the primal form, since $x=(r,A)\in K_1$ if and only if $r\geq 0$ and $A\geq 0$, and $v_2-\cT(x)\in K_1$ if and only if $\rho_2-r\rho_1+A\geq 0$, we obtain
	\be
	\begin{split}
		&\max_{\substack{x\in K_1\\
				v_2-\cT(x)\in K_2}}\langle v_1,x\rangle_1\\
			&=\max_{\substack{A\geq r\rho_1-\rho_2\\
				r,A\geq 0}}\left\{(1-\epsilon)r-\tr[A]\right\}.
	\end{split}
	\ee
	The right-hand side of the above equation can be further simplified as follows. We first fix $r$ and optimize over $A$. Since the $A\ge0$ with minimum trace such that $A\ge r\rho_1-\rho_2$ is exactly $A=(r\rho_1-\rho_2)_+$,
	we conclude that
	\[
	\max_{\substack{x\in K_1\\
			v_2-\cT(x)\in K_2}}\langle v_1,x\rangle_1=\max_{r\geq 0}\left\{(1-\epsilon)r-\tr(r\rho_1-\rho_2)_+\right\},
	\]
	that is,
	\be\label{eq:spoiler1}
	Q^\epsilon(\rho_1\|\rho_2)=\max_{r\geq 0}f_\epsilon(r),
	\ee
	where
	\be\label{eq:spoiler2}
	f_\epsilon(r)\defeq(1-\epsilon)r-\tr(r\rho_1-\rho_2)_+\;.
	\ee
	Note that
	\[
	\tr(r\rho_1-\rho_2)_+=\frac{\N{r\rho_1-\rho_2}_1+r-1}{2}\;,
	\]
	that is,
	\[
	f_\epsilon(r)=\frac{1+(1-2\epsilon)r-\N{r\rho_1-\rho_2}_1}{2}.
	\]
	Therefore, denoting $f'_\epsilon(r)= 2^{-1}\{1+(1-2\epsilon)r-\N{r\rho'_1-\rho'_2}_1\}$, we have that
	\[
	\N{r\rho_1-\rho_2}_1\ge \N{r\rho'_1-\rho'_2}_1\implies f_\epsilon(r)\le f'_\epsilon(r),
	\]
	independently of $r$ and $\epsilon$. We thus have proved that (ii) implies (i).
	
	To show that (i) implies (ii), suppose 
	$Q^\epsilon(\rho_1\|\rho_2)\leq Q^\epsilon(\rho'_1\|\rho'_2)$ for all $\epsilon\in[0,1]$. Let $r_\epsilon\ge 0$ be the minimum value of $r$ achieving $Q^\epsilon(\rho'_1\|\rho'_2)$, in formula,
	\begin{equation}\label{eq:r-eps}
	r_\epsilon\defeq\min\{r\ge0: f'_\epsilon(r)=Q^\epsilon(\rho'_1\|\rho'_2)\}.
	\end{equation}
	In all such points $r_\epsilon$, definition~(\ref{eq:spoiler1}) together with the assumption (i) guarantee that $\N{r_\epsilon\rho_1-\rho_2}_1\ge\N{r_\epsilon\rho'_1-\rho'_2}_1$. This fact can be simply shown by the following chain of inequalities:
	\begin{equation*}
	\begin{split}
		&\frac{1+(1-2\epsilon)r_\epsilon-\N{r_\epsilon\rho'_1-\rho_2'}_1}{2}\\
		&=Q^\epsilon(\rho'_1\|\rho'_2)\\
		&\ge Q^\epsilon(\rho_1\|\rho_2)\\
		&=\max_r\frac{1+(1-2\epsilon)r-\N{r\rho_1-\rho_2}_1}{2}\\
		&\ge \frac{1+(1-2\epsilon)r_\epsilon-\N{r_\epsilon\rho_1-\rho_2}_1}{2}.
	\end{split}
	\end{equation*}
	The crucial observation now is that the points $r_\epsilon$, representing the solutions of~(\ref{eq:r-eps}) for varying $\epsilon\in[0,1]$, coincide with the points where the quantity $\N{r\rho'_1-\rho'_2}_1$, thought as a function of $r$, changes its slope (see Fig.~\ref{fig:plot} below). For example, for $\epsilon=0$, we have to consider the function
	\[
	f'_0(r)=\frac{1+r-\N{r\rho'_1-\rho'_2}_1}{2},
	\]
	and this achieves its maximum value 1 for $r\ge r^*\equiv r_0=\sup(\rho'_2/\rho'_1)$. But then, if we know that the curve $\N{r\rho_1-\rho_2}_1$ is not below $\N{r\rho'_1-\rho'_2}_1$ in all the points where the latter changes its slope, this is sufficient to conclude that
	\[
	\N{r\rho'_1-\rho'_2}_1\le \N{r\rho_1-\rho_2}_1,\quad\forall r\ge 0.
	\]
	This completes the proof of the lemma.
\end{proof}
	\begin{figure}[htb]
		\centering
		\includegraphics[width=0.9\linewidth]{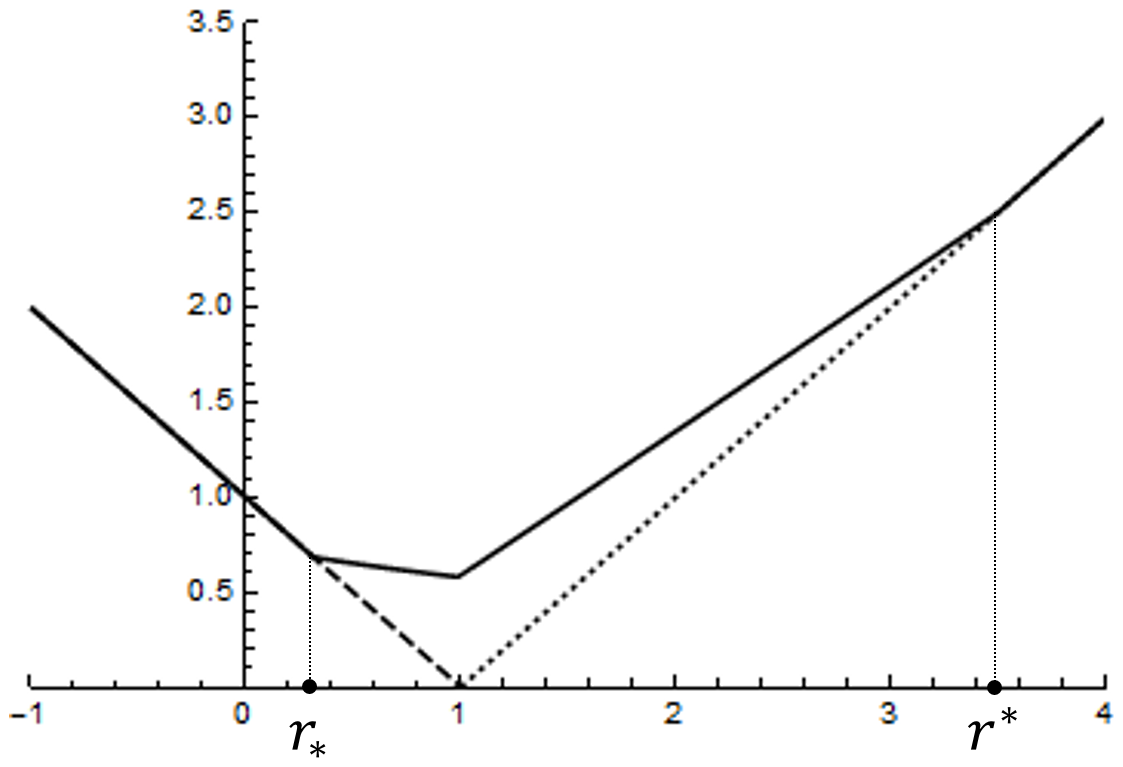}
		\caption{Typical behavior of $\N{r\rho_1-\rho_2}_1$, for two random density matrices $\rho_1$ and $\rho_2$ on $\mathbb{C}^3$, as a function of $r\in\mathbb{R}$ (continuous line). For $r\le r_*\equiv\inf(\rho_2/\rho_1)=\sup\{\lambda:\lambda\rho_1-\rho_2\le 0 \}$, the curve becomes equal to $1-r$ (dashed line). For $r\ge r^*\equiv\sup(\rho_2/\rho_1)=\inf\{\lambda : \lambda\rho_1-\rho_2\ge 0\}$, the curve becomes equal to $r-1$ (dotted line).}
		\label{fig:plot}
	\end{figure}

%
%

\subsubsection{The classical case}

As a ``consistency check'' we separately consider the classical case here. Suppose $\rho$ and $\sigma$ are both diagonal with elements
$p_1,...,p_n$ and $q_1,...,q_n$, respectively. Denote $r_j\equiv q_j/p_j$ if $p_j>0$ and otherwise $r_j=0$. W.l.o.g. suppose
$r_1\leq r_2\leq\cdots\leq r_n$. We therefore get
\be
rp_j-q_j> 0\;\;\iff\;\;r> r_j
\ee
Hence, for $r\in(r_{k},r_{k+1}]$
\be
f(r)=(1-\epsilon)r-\tr(r\rho-\sigma)_+=(1-\epsilon)r-r\sum_{j=1}^{k}p_j+\sum_{j=1}^{k}q_j
\ee
Due to the linearity in $r$ of the expression above 
we conclude that in the classical case
\begin{align*}
Q^\epsilon(\rho\|\sigma)&=\max_{k\in\{1,...,n\}}f(r_k)\\
&=\max_{k\in\{1,...,n\}}\left\{\left(1-\epsilon-\sum_{j=1}^{k}p_j\right)r_k+\sum_{j=1}^{k}q_j\right\},
\end{align*}
thus reconstructing the Blackwell criterion for pairs of probability distributions (including majorization and thermomajorization).

\section{applications}

In this section, we study how the conditions in Theorem~\ref{theo:1} are logically related to the existence of a suitable transformation mapping $(\rho_1,\rho_2)$ into $(\rho_1',\rho'_2)$.

\subsection{Coherent energy transitions with Gibbs-preserving operations}

We consider here the resource theory of 
athermality~\cite{horodecki_fundamental_2013,brandao_resource_2013}. In this theory, quantum systems that are not in thermal equilibrium with their environment are considered resources (e.g., work can be extracted from such systems). Hence, free systems are those prepared in the Gibbs state, i.e., $\gamma=\mathcal{Z}^{-1}\sum_{x=1}^{d}e^{-\beta E_x}|x\>\<x|$, and permitted operations must preserve $\gamma$. Consider now two possibly non-commuting quantum states $\rho$ and $\sigma$ both with same 2-dimensional support spanned by the same two energy eigenstates, say $|x\>$ and $|y\>$. Such states form the building blocks of quantum thermodynamics as they contain the smallest units of athermality~\cite{brandao_resource_2013}. Here we find both necessary and sufficient conditions under which a Gibbs-preserving transition between $\rho$ and $\sigma$ is possible. We call  such transitions {\it coherent energy transitions}, since not only the transitions $|x\>\to|y\>$ and $|y\>\to|x\>$ are considered, but also transitions between any linear superpositions of such energy eigenstates.

The main result about coherent energy transitions is the following:
\begin{theorem}\label{theo:qubit-gibbs}
	With $\gamma$, $\rho$, and $\sigma$ as above, assuming that $\gamma>0$ (i.e., non-zero temperature), the following are equivalent:
	\begin{enumerate}[label=\roman*)]
		\item $\rho$ can be transformed into $\sigma$ by a $\gamma$-preserving CPTP operation (i.e., Gibbs-preserving operation);
		\item $(\rho,\gamma)\succ(\sigma,\gamma)$;
		\item \begin{equation}\label{cond-simple}
		\begin{cases}
		D_{\max}(\rho\|\gamma)\geq D_{\max}(\sigma\|\gamma),\\
		D_{\max}(\gamma\|\rho)\geq D_{\max}(\gamma\|\sigma).
		\end{cases}
		\end{equation}
	\end{enumerate} 	
\end{theorem}
In other words, in this case, we do not need to check the validity of point~(iii) of Theorem~\ref{theo:1} for all values of $\alpha$, but only in the limit $\alpha\to\infty$. Moreover, in this case, we know that a CPTP map between the two pairs of states exists. In the case of zero temperature, i.e., if $\gamma\ge 0$, a third condition has to be added to the above list, namely, $D_{\min}(\gamma\|\rho)\geq D_{\min}(\gamma\|\sigma)$,
where $D_{\min}(\gamma\|\rho)=-\log\Tr[\Pi_\gamma\ \rho]$ denotes the \textit{min-relative entropy}~\cite{datta_min-_2009} and $\Pi_\gamma$ is the projector onto the support of $\gamma$. Finally, we did not include the condition $D_{\min}(\rho\|\gamma)\geq D_{\min}(\sigma\|\gamma)$ since it is trivial, unless $\sigma$ is rank-one (i.e., a pure state). However, as shown below in the proof, it turns out that in this case the other conditions implies this one.

Theorem~\ref{theo:qubit-gibbs} generalizes an earlier work given in~\cite{Varun2015} to the generic case in which the Gibbs state is not pure. It demonstrates that three athermality monotones (given in terms of the min/max relative entropies) provide both necessary and sufficient conditions for the existence of a Gibbs preserving map connecting two non-thermal states with the same two-dimensional support. Since the set of Gibbs-preserving operations is strictly larger than the set of thermal-operations~\cite{brandao_resource_2013}, these three monotones, in general, will not be sufficient to determine convertibility under thermal operations~\cite{Hor15}.

The Gibbs state is given by
\be\label{g}
\gamma=\frac{1}{\mathcal{Z}}\sum_{x=1}^{d}\exp(-\beta E_x)|x\rangle\langle x|,
\ee
where $\beta=1/kT$ is the inverse temperature, $d$ is the dimension of the quantum system, $\{|x\rangle\}_{x=1}^{d}$ are the complete set of eigenstates of the Hamiltonian, $E_x$ the eigenvalues of the Hamiltonian, and $\mathcal{Z}$ is the partition function $\sum_{x=1}^{d}\exp(-\beta E_x)$. Consider now two quantum states $\rho$ and $\sigma$ both with the same 2-dimensional support given by, e.g.,
\be\label{rank2}
{\rm supp}(\rho)={\rm supp}(\sigma)={\rm span}\left\{|x\rangle\right\}_{x=1,2},
\ee
but it does not matter which two energy eigenstates are chosen (with the condition $E_2\ge E_1$).
Denote further by $\gamma^{(2)}$ the Gibbs state projected onto this two dimensional subspace: 
\be\label{g2}
\gamma^{(2)}\equiv p|1\rangle\langle 1|+(1-p)|2\rangle\langle 2|
\ee
where
\be
p\equiv \frac{\exp(-\beta E_1)}{\exp(-\beta E_1)+\exp(-\beta E_2)}=\frac{1}{1+\exp(-\beta\Delta E)}\;,
\ee
with $\Delta E=E_2-E_1\geq 0$ so that $p\geq 1/2$. We start with the following lemma:
  \begin{lemma}
  Let $\rho$ and $\sigma$ be as in~\eqref{rank2}, and let $\gamma$ be as in~\eqref{g} with $\gamma^{(2)}$ as in~\eqref{g2}. Then,
  there exists a CPTP map $\Phi$ such that $\Phi(\rho)=\sigma$ and $\Phi(\gamma)=\gamma$ if and only if there exists a CPTP map $\mE$ such that $\mE(\rho)=\sigma$ and $\mE(\gamma^{(2)})=\gamma^{(2)}$.
  \end{lemma}
  
  \begin{proof}
  Suppose there exists $\Phi$ such that $\Phi(\rho)=\sigma$ and $\Phi(\gamma)=\gamma$. Then, for all
  $t>0$ we have
  \be\label{AU}
  \N{\sigma-t\gamma}_1=\N{\Phi(\rho)-\Phi(\gamma)}_1\leq \N{\rho-t\gamma}_1\;,
  \ee
  since the trace norm is contractive. 
  Next, denoting by $P$ the projection onto $\text{span}\{|1\>,|2\>\}$ and  $r\equiv \exp(-\beta E_1)+\exp(-\beta E_2)$, we have
\[
\begin{split}
  \N{\rho-t\gamma}_1&=\N{P\rho P-tP\gamma P}_1+t\N{(I-P)\gamma(I-P)}_1\\
  &=\N{\rho-t\frac{r}{\mathcal{Z}}\gamma^{(2)}}_1+t\N{\gamma-\frac{r}{\mathcal{Z}}\gamma^{(2)}}_1,
\end{split}
\]
and, analogously, 
\[
\N{\sigma-t\gamma}_1=\N{\sigma-t\frac{r}{\mathcal{Z}}\gamma^{(2)}}_1+t\N{\gamma-\frac{r}{\mathcal{Z}}\gamma^{(2)}}_1.
\]  
%
  Therefore, since $r/\mathcal{Z}>0$, we can introduce the new parameter $t'\defeq t\cdot\frac{r}{\mathcal{Z}}$ so that the inequality in Eq.~\eqref{AU} can be rewritten as
  \be
  \N{\sigma-t'\gamma^{(2)}}_1\leq \N{\rho-t'\gamma^{(2)}}_1\;\;\;\forall\;t'>0\;.
  \ee
From the Alberti--Uhlmann result on qubits~\cite{alberti_problem_1980} there exists $\mE$ as in the lemma.
  
  Conversely, suppose there exists a CPTP map $\mE$ such that $\mE(\rho)=\sigma$ 
  and $\mE(\gamma^{(2)})=\gamma^{(2)}$. Then, define $\Phi$ as follows. Let $P=|1\ra\la 1|+|2\ra\la 2|$ be the projector onto the support of $\rho$ and $\sigma$, and define
  \be
  \Phi(\cdot):=\mE\left(P(\cdot) P\right)+(I_d-P)(\cdot)(I_d-P)\;.
  \ee
  By construction, $\Phi$ is CPTP since $\mE$ is CPTP, and it is easy to verify that $\Phi(\rho)=\sigma$ and $\Phi(\gamma)=\gamma$. This completes the proof.
  \end{proof}

\begin{lemma}
	Let $\rho$ and $\sigma$ be two qubit density matrices and let $\gamma^{(2)}$ be the Gibbs state given in~\eqref{g2}. 
	Then, $\rho$ can be converted to $\sigma$ by Gibbs preserving operations if and only if the following three inequalities simultaneously hold:
	\begin{equation}\label{cond}
	\begin{cases}
	D_{\max}(\rho\|\gamma^{(2)})\geq D_{\max}(\sigma\|\gamma^{(2)}),\\
	D_{\max}(\gamma^{(2)}\|\rho)\geq D_{\max}(\gamma^{(2)}\|\sigma),\\
	D_{\min}(\gamma^{(2)}\|\rho)\geq D_{\min}(\gamma^{(2)}\|\sigma).
	\end{cases}
	\end{equation}
\end{lemma}

Before proceeding, we notice that, while in Eq.~(\ref{cond}) above the projected Gibbs state $\gamma^{(2)}$ appears, in Eq.~(\ref{cond-simple}) of Theorem~\ref{theo:qubit-gibbs} we use the original $\gamma$. However, since $P\rho P=\rho$, $P\sigma P=\sigma$, and $P\gamma P=c\gamma^{(2)}$ (for some $c\ge0$), and since both $D_{\max}$ and $D_{\min}$ in this case only depend on what there is on the support of $P$, the two set of conditions are clearly equivalent.

	Denote by
	\begin{align}
	& m(\rho,\gamma^{(2)})\defeq\inf(\rho/\gamma^{(2)})=\sup\{t\in\mbb{R}\;:\;t\gamma^{(2)}-\rho\leq 0\}\nonumber\\
	& M(\rho,\gamma^{(2)})\defeq\sup(\rho/\gamma^{(2)})=\inf\{t\in\mbb{R}\;:\;t\gamma^{(2)}-\rho\geq 0\}\;.
	\end{align}
	Note that $m(\rho,\gamma^{(2)})\leq 1\leq M(\rho,\gamma^{(2)})$. Since we consider here the qubit case, it follows that 
	$m(\rho,\gamma^{(2)})$ and $M(\rho,\gamma^{(2)})$ are the roots to the quadratic polynomial Det$(\rho-t\gamma^{(2)})$.
	A straightforward calculation gives (assuming $\det(\gamma^{(2)})>0$)
	\be
	\det(\rho-t\gamma^{(2)})=\det(\gamma^{(2)})\left[t-m(\rho,\gamma^{(2)})\right]\left[t-M(\rho,\gamma^{(2)})\right]
	\ee
	with $m$ and $M$ given explicitly below after we introduce a few notations.
	
	Without loss of generality we can assume that the off-diagonal terms of $\rho$ are non-negative real numbers since $\gamma^{(2)}$ is invariant under conjugation by any $2\times2$ unitary matrix which is diagonal on the energy eigenbasis (i.e., commutes with $\gamma^{(2)}$). Hence, we can write
	\be\label{rho}
	\rho=\begin{pmatrix}a & \varepsilon\sqrt{a(1-a)}\\ \varepsilon\sqrt{a(1-a)} & 1-a\end{pmatrix}
	\ee 
	with $\varepsilon,a\in[0,1]$.
	Taking $\gamma^{(2)}$ as in~\eqref{g2} we get the following explicit expressions for
	$m(\rho,\gamma^{(2)})$ and $M(\rho,\gamma^{(2)})$ assuming $\det(\gamma^{(2)})>0$ (i.e. $0<p<1$) 
	\begin{align}
	& m(\rho,\gamma^{(2)})=\frac{1}{2}\left[r_0+r_1-\sqrt{(r_0-r_1)^2+4r_0r_1\epsilon^2}\right]\nonumber\\
	& M(\rho,\gamma^{(2)})=\frac{1}{2}\left[r_0+r_1+\sqrt{(r_0-r_1)^2+4r_0r_1\epsilon^2}\right]
	\end{align}
	where
	\be
	r_0\equiv \frac{a}{p}\quad\text{and}\quad r_1\equiv\frac{1-a}{1-p}\;.
	\ee
	
	By definition, both $m$ and $M$ are monotonic in the sense that 
	\begin{align}
	M\left(\Phi(\rho),\Phi(\gamma^{(2)})\right) & \leq M(\rho,\gamma^{(2)})\nonumber\\
	m\left(\Phi(\rho),\Phi(\gamma^{(2)})\right)& \geq m(\rho,\gamma^{(2)})\;.
	\end{align}
	Note that $m(\rho,\gamma^{(2)})=1/M(\gamma^{(2)},\rho)$ and $M$ is related to the max relative entropy:
	\begin{align}
	&D_{\max}(\rho\|\gamma^{(2)})=\log M(\rho,\gamma^{(2)})\nonumber\\
	&D_{\max}(\gamma^{(2)}\|\rho)=\log M(\gamma^{(2)},\rho)=-\log m(\rho,\gamma^{(2)})
	\end{align}
	A dual definition is the min-relative entropy defined by:
	\be
	D_{\min}(\gamma^{(2)}\|\rho)=-\log\tr\left[\rho\Pi_{\gamma^{(2)}}\right]\;,
	\ee
	where $\Pi_{\gamma^{(2)}}$ is the projection to the support of $\gamma^{(2)}$.
	Clearly, the if $\det(\gamma^{(2)})>0$ then $D_{\min}(\gamma^{(2)},\rho)=0$. Summarizing, we showed that the conditions in Eq.~(\ref{cond}) are equivalent to (remember the assumption here $\gamma^{(2)}>0$; the case of rank-one $\gamma^{(2)}$ will be considered separately below)
	\[
	\begin{cases}
	M(\rho,\gamma^{(2)})\ge M(\sigma,\gamma^{(2)})\\
	m(\rho,\gamma^{(2)})\le m(\sigma,\gamma^{(2)}).
	\end{cases}
	\]

To see how the above conditions can be used to prove Theorem~\ref{theo:qubit-gibbs}, we need the following lemma from Ref.~\cite{alberti_problem_1980}:
	\begin{lemma}(Alberti--Uhlmann)
		Let $\rho$, $\sigma$, $\eta$, and $\tau$ be qubit density matrices.
		Then, there exists a CPTP map $\Phi$ such that $\sigma=\Phi(\rho)$ and $\eta=\Phi(\tau)$ if and only if
		\be\label{condd}
		M(\rho,\tau)\geq M(\sigma,\eta)\geq m(\sigma,\eta)\geq  m(\rho,\tau)\;,
		\ee
		and
		\be\label{redan}
		\det(\sigma-t\eta)\geq \det(\rho-t\tau)\;\;\forall\;\;m(\sigma,\eta)\leq t \leq M(\sigma,\eta)\;.
		\ee
	\end{lemma}
	
We now apply the above lemma above to the case $\tau=\eta=\gamma^{(2)}$.

\subsubsection{First case: non-zero temperature ($\gamma^{(2)}>0$).}
		We first assume $\det(\gamma^{(2)})>0$.
		The necessity of~\eqref{cond} follows from the fact that the min/max relative entropies both satisfies the data processing inequality. We therefore need to show that they are sufficient.
		With the choice $\tau=\eta=\gamma^{(2)}$ the conditions~\eqref{condd} are equivalent to the conditions~\eqref{cond}
		(recall the last condition of~\eqref{cond} is trivial since we assume for now that $\gamma^{(2)}$ is full rank).
		It is therefore left to show that the conditions~\eqref{redan} hold automatically if Eqs.~\eqref{cond} hold.
		Indeed, recall that for $\det(\gamma^{(2)})>0$ we have
		\be
		\det(\rho-t\gamma^{(2)})=\det(\gamma^{(2)})\left[t-m(\rho,\gamma^{(2)})\right]\left[t-M(\rho,\gamma^{(2)})\right]
		\ee
		and
		\be
		\det(\sigma-t\gamma^{(2)})=\det(\gamma^{(2)})\left[t-m(\sigma,\gamma^{(2)})\right]\left[t-M(\sigma,\gamma^{(2)})\right]
		\ee
		Hence, the inequality $\det(\sigma-t\gamma^{(2)})\geq \det(\rho-t\gamma^{(2)})$ is equivalent to
		\begin{align}\label{calc}
		&t\left[M(\rho,\gamma^{(2)})+m(\rho,\gamma^{(2)})-M(\sigma,\gamma^{(2)})-m(\sigma,\gamma^{(2)})\right]\nonumber\\
		&+m(\sigma,\gamma^{(2)})M(\sigma,\gamma^{(2)})-m(\rho,\gamma^{(2)})M(\rho,\gamma^{(2)})\geq 0
		\end{align}
		We therefore need to show that the above inequality holds for all $m(\sigma,\gamma^{(2)})\leq t \leq M(\sigma,\gamma^{(2)})$.
		It is therefore sufficient to show that it holds at the two extreme points of the interval. Indeed, for $t=m(\sigma,\gamma^{(2)})$
		after some algebra the expression in~\eqref{calc} becomes
		$$
		\left[M(\rho,\gamma^{(2)})-m(\sigma,\gamma^{(2)})\right]\left[m(\sigma,\gamma^{(2)})-m(\rho,\gamma^{(2)})\right]
		$$
		which is non-negative due to~\eqref{cond}. Similarly, substituting $t=M(\sigma,\gamma^{(2)})$ in~\eqref{calc} gives
		$$
		\left[M(\sigma,\gamma^{(2)})-m(\rho,\gamma^{(2)})\right]\left[M(\rho,\gamma^{(2)})-M(\sigma,\gamma^{(2)})\right]
		$$
		which is again non-negative due to~\eqref{cond}. This completes the proof of the theorem for the case $\det(\gamma^{(2)})>0$.
		
		\subsubsection{Second case: zero temperature ($\det(\gamma^{(2)})=0$).} In this case direct calculation gives
		\be
		\det(\rho-t\gamma^{(2)})=\det(\rho)-t\left(1-\tr[\rho\Pi_{\gamma^{(2)}}]\right)
		\ee
		and
		\be
		\det(\sigma-t\gamma^{(2)})=\det(\sigma)-t\left(1-\tr[\sigma\Pi_{\gamma^{(2)}}]\right).
		\ee
		Note that $\Pi_{\gamma^{(2)}}$ is either the projection $|1\ra\la1|$ or $|2\ra\la 2|$. 
		
		Now, assuming $\rho\neq\sigma\neq\gamma^{(2)}$ (otherwise the problem becomes trivial), since $\gamma^{(2)}$ is rank 1 we have $M(\rho,\gamma^{(2)})=M(\sigma,\gamma^{(2)})=\infty$. On the other hand, in this case a simple calculation gives
		\begin{align*}
		m(\rho,\gamma^{(2)})&=\frac{\det(\rho)}{1-\tr[\rho\Pi_{\gamma^{(2)}}]}\\ 
		m(\sigma,\gamma^{(2)})&=\frac{\det(\sigma)}{1-\tr[\sigma\Pi_{\gamma^{(2)}}]}
		\end{align*}
		where we used $\rho\neq\gamma^{(2)}$ and $\sigma\neq\gamma^{(2)}$, so that together with $\gamma^{(2)}$ being rank 1 gives $\tr[\rho\Pi_{\gamma^{(2)}}]<1$ and similarly $\tr[\sigma\Pi_{\gamma^{(2)}}]<1$.
		Hence, in this case exploring the behaviour of $\det(\sigma-t\gamma^{(2)})\geq \det(\rho-t\gamma^{(2)})$ in the limit $t\to M(\sigma,\gamma^{(2)})=\infty$ we must have 
		$$
		\tr[\sigma\Pi_{\gamma^{(2)}}]\geq \tr[\rho\Pi_{\gamma^{(2)}}]
		$$
		which is equivalent to $D_{\min}(\gamma^{(2)}\|\rho)\geq D_{\min}(\gamma^{(2)}\|\sigma)$. At the point $t=m(\sigma,\gamma^{(2)})$, the inequality $\det(\sigma-t\gamma^{(2)})\geq \det(\rho-t\gamma^{(2)})$ becomes
		\begin{align*}
		0 & \geq \det(\rho)-m(\sigma,\gamma^{(2)})\left(1-\tr[\rho\Pi_{\gamma^{(2)}}]\right)\\
		&=\left(1-\tr[\rho\Pi_{\gamma^{(2)}}]\right)\left(\frac{\det(\rho)}{1-\tr[\rho\Pi_{\gamma^{(2)}}]}-m(\sigma,\gamma^{(2)})\right)\\
		&=\left(1-\tr[\rho\Pi_{\gamma^{(2)}}]\right)\Big(m(\rho,\gamma^{(2)})-m(\sigma,\gamma^{(2)})\Big)
		\end{align*}
		which is satisfied since $m(\rho,\gamma^{(2)})\leq m(\sigma,\gamma^{(2)})$.
		Hence, $\det(\sigma-t\gamma^{(2)})\geq \det(\rho-t\gamma^{(2)})$
		for all $t$ with $m(\sigma,\gamma^{(2)})\le t<\infty$.
		This completes the proof.

\subsection{Test-and-prepare channels}

The proof of Theorem~\ref{theo:qubit-gibbs} above relies on a lemma proved by Alberti and Uhlmann~\cite{alberti_problem_1980}, which, together with Theorem~\ref{theo:1}, implies that, if $n=m=2$ (i.e., for qubits) then $(\rho_1,\rho_2)\succ(\rho'_1,\rho'_2)$ if and only if there exists a completely positive trace-preserving (CPTP) map $\Phi$ such that $\Phi(\rho_i)=\rho'_i$ ($i=1,2$). However, explicit counterexamples exist, showing that as soon as one leaves the qubit case, already when $n=3$ and $m=2$, this is not true anymore~\cite{matsumoto_example_2014}. Hence, leaving aside the general case, we focus instead on a special class of CPTP maps, namely, \textit{test-and-prepare channels} of the form:
\[
\mE(\rho)\defeq\Tr[E\rho]\xi_1+\Tr[(\openone-E)\rho]\xi_2,
\]
for some effect $0\le E\le\openone$ and some density matrices $\xi_1$ and $\xi_2$. Test-and-prepare channels are, in other words, measure-and-prepare channels for which the measurement has only two possible outcomes. Although restricted, this class seems quite natural in the framework of quantum relative Lorenz curves, which are defined only in terms of binary measurements (i.e., hypothesis tests). Indeed, a necessary and sufficient condition for the existence of a test-and-prepare channel between two pairs of density matrices can be expressed in terms of quantum relative Lorenz curves as follows:
\begin{theorem}\label{theo:meas-and-prep}
	Given two pairs of density matrices $(\rho_1,\rho_2)$ and $(\rho'_1,\rho'_2)$ on $\mathbb{C}^n$ and $\mathbb{C}^m$, respectively, there exists a test-and-prepare channel $\mE$ such that $\mE(\rho_1)=\rho_1'$ and $\mE(\rho_2)=\rho_2'$, if and only if the quantum Lorenz curve of $\rho_1$ relative to $\rho_2$ is nowhere below the segments joining the points $(0,0)$, $(1,1)$ and passing through either \[(x,y)=\left(\frac{1-m'}{M'-m'},\frac{M'(1-m')}{M'-m'}\right)\] or \[(x,y)=\left(\frac{M'-1}{M'-m'},\frac{m'(M'-1)}{M'-m'}\right),\] whichever is higher,
	where \[M'\defeq 2^{D_{\max}(\rho'_1\|\rho_2')}=\sup(\rho_1'/\rho_2')\] and \[m'\defeq2^{-D_{\max}(\rho'_2\|\rho'_1)}=\inf(\rho_1'/\rho_2').\]
\end{theorem}
\begin{proof}
Consider a test-and-prepare channel of the form:
\be\label{mpc}
\mE(\rho)=\Tr[E\rho]\sigma_1+\Tr\left[(\openone-E)\rho\right]\sigma_2
\ee
where $\sigma_1,\sigma_2$ are density matrices (i.e. positivie semi-definite matrices with trace 1) and 
$0\leq E\leq \openone$. If $\mE(\rho_j)=\rho_{j}'$ for $j=1,2$, then
\begin{align*}
\rho_1'=e_1\sigma_1+(1-e_1)\sigma_2\\
\rho_2'=e_2\sigma_1+(1-e_2)\sigma_2
\end{align*}
where $e_j\equiv\Tr[E\rho_j]$ for $j=1,2$. Assuming $e_1\neq e_2$ (otherwise, $\rho_1'=\rho_2'$), the above equations are equivalent to
\begin{align*}
\sigma_1&=\frac{1}{e_1-e_2}\left[(1-e_2)\rho_1'-(1-e_1)\rho_2'\right]\\
\sigma_2&=\frac{1}{e_1-e_2}\left[-e_2\rho_1'+e_1\rho_2'\right]
\end{align*} 
Note that $\sigma_1$ and $\sigma_2$ have trace $1$ since $\rho_1'$ and $\rho_2'$ have trace 1.
W.l.o.g. we can assume $e_1>e_2$. With this choice, $\sigma_1$ and $\sigma_2$ are positive semi-definite if and only if
\be
\rho_1'-\frac{1-e_1}{1-e_2}\rho_2'\geq 0\quad\text{and}\quad
\rho_2'-\frac{e_2}{e_1}\rho_1'\geq 0\;.
\ee
Note that since we assume $e_1>e_2$ we have $1-e_2>0$ and $e_1>0$.
Denote by $m'\equiv\inf(\rho_1'/\rho_2')$ and by $M'\equiv\sup(\rho_1'/\rho_2')$, and note that
$\inf(\rho_2'/\rho_1')=1/M'$. We therefore get that $\sigma_1$ and $\sigma_2$ are positive semi-definite if and only if
\be
\frac{1-e_1}{1-e_2}\leq m' \quad\text{and}\quad
\frac{e_2}{e_1}\leq \frac{1}{M'}\;.
\ee
The above inequalities are equivalent to
$$
\begin{cases}
\Tr\left[E\left(\rho_1-m'\rho_2\right)\right]\geq 1-m'\;,\\
\Tr\left[E\left(\rho_1-M'\rho_2\right)\right]\geq 0\;.
\end{cases}
$$
We therefore arrive at the following lemma:
\begin{lemma}
	There exists a channel $\mE$ of the form~\eqref{mpc} such that $\mE(\rho_1)=\rho_1'$ and $\mE(\rho_2)=\rho_2'$, if and only if
	\be
	W(\rho_1,\rho_2,\rho_1',\rho_2')\geq 0
	\ee 
	where the witness $W$ is defined as
	\begin{align}
	W(&\rho_1,\rho_2,\rho_1',\rho_2')\defeq \nonumber\\
	&m'-1+\max_{\substack{0\leq E\leq \openone\\
			\Tr\left[E\left(\rho_1-M'\rho_2\right)\right]\geq 0}}\Tr\left[E\left(\rho_1-m'\rho_2\right)\right]\label{eq:witness}	
	\end{align}
\end{lemma}
The calculation of $W$ can be simplified using the following dual formulation of linear programming, analogously to what we did in the proof of Lemma~\ref{lem:5}. Let $V_1$ and $V_2$ be two (inner product) vector spaces with two cones $K_1\subset V_1$ and $K_2\subset V_2$. Consider two vectors $v_1\in V_1$ and $v_2\in V_2$, and a linear map
$\cT:V_1\to V_2$. Then, the primal form is:
\be\label{sdp1}
\max_{\substack{x\in K_1\\
		v_2-\cT(x)\in K_2}}\langle v_1,x\rangle_1
\ee
The dual form involves the adjoint map $\cT^*:V_2\to V_1$:
\be\label{sdp2}
\min_{\substack{y\in K_2\\
		\cT^*(y)-v_1\in K_1}}\langle v_2,y\rangle_2
\ee
For our purposes, we take $V_1=\mH_n$ the space of $n\times n$ Hermitian matrices, and we take $K_1=\mH_{n,+}$ the cone of positive semi-definite matrices in $\mH_n$. We further define the vector space
\be
V_2\equiv \mbb{R}\oplus \mH_n=\left\{(r,A)\;\Big|\;r\in\mbb{R}\;\;;\;\;A\in\mH_n\right\}\;,
\ee
with inner product $\left\langle (r,A),(t,B)\right\rangle_1:=rt+\tr[AB]$.
Further, define $K_2=\mbb{R}_{+}\oplus\mH_{n,+}$ to be the positive cone in $V_2$.
The linear map $\cT:V_1\to V_2$ is defined as follows:
\be
\cT(A)=\left(-\Tr\left[A\left(\rho_1-M'\rho_2\right)\right],A\right)\;.
\ee
Note that the dual map $\cT^{*}:V_2\to V_1$ is given by
\be
\cT^*(r,A)=A-r\left(\rho_1-M'\rho_2\right)\;.
\ee
Finally, set $v_1=\rho_1-m'\rho_2$ and $v_2=(0,\openone_n)$. With these choices, $v_2-\cT(x)=\left(\Tr\left[x\left(\rho_1-M'\rho_2\right)\right],\openone_n-x\right)$, so that the primal problem becomes
\be
\begin{split}
	&\max_{\substack{x\in K_1\\
			v_2-\cT(x)\in K_2}}\langle v_1,x\rangle_1\\
		&=\max_{\substack{0\leq E\leq \openone_n\\
			\Tr\left[E\left(\rho_1-M'\rho_2\right)\right]\geq 0}}\Tr\left[E\left(\rho_1-m'\rho_2\right)\right]
\end{split}
\ee
where we renamed $x$ with $E$. The dual problem is given by
\be
\begin{split}
&\min_{\substack{y\in K_2\\
		\cT^*(y)-v_1\in K_1}}\langle v_2,y\rangle_2\\
	&=
\min_{\substack{r,F\geq 0\\
		F\geq r(\rho_1-M'\rho_2)+(\rho_1-m'\rho_2)}}\Tr[F]
\end{split}
\ee
where we took $y=(r,F)$. We can further simplify the above expression. First note that, for any given $r$,
the positive semi-definite matrix $F$ with the smallest trace that satisfies
\[
\begin{split}
F&\geq  r(\rho_1-M'\rho_2)+(\rho_1-m'\rho_2)\\
&=(1+r)\rho_1-(rM'+m')\rho_2
\end{split}
\]
is of course the positive part of the left-hand side:
\be
F=\left[(1+r)\rho_1-(rM'+m')\rho_2\right]_+\;.
\ee
We therefore conclude that the dual problem is equivalent to
\begin{align}
&\min_{r\geq 0}\Tr\left[(1+r)\rho_1-(rM'+m')\rho_2\right]_+\nonumber\\
&=\min_{r\geq 0}\frac{1-m'-r(M'-1)+\N{(1+r)\rho_1-(rM'+m')\rho_2}_1}{2}\;.
\end{align}
By plugging the above equation into~(\ref{eq:witness}), we therefore conclude that
\begin{equation}
\begin{split}
&W(\rho_1,\rho_2,\rho_1',\rho_2')
=m'-1+\frac{1}{2}(1-m')\nonumber\\
&+\frac{1}{2}\min_{r\geq 0}\Big(\N{(1+r)\rho_1-(rM'+m')\rho_2}_1-r(M'-1)\Big)\nonumber\\
&=\frac12(m'-1)\nonumber\\
&+\frac{1}{2}\min_{r\geq 0}\Big(\N{(1+r)\rho_1-(rM'+m')\rho_2}_1-r(M'-1)\Big),
\end{split}
\end{equation}
namely,
\begin{align}
&2W(\rho_1,\rho_2,\rho_1',\rho_2')=-(1-m')\nonumber\\
&+\min_{r\geq 0}\Big(\N{(1+r)\rho_1-(rM'+m')\rho_2}_1-r(M'-1)\Big)\;.
\end{align}
Introducing $$t\defeq\frac{rM'+m'}{1+r}\;,$$ and noting that $t\in[m',M')$ we obtain that there exists a channel $\mE$ of the form~\eqref{mpc} such that $\mE(\rho_1)=\rho_1'$ and $\mE(\rho_2)=\rho_2'$, if and only if
\begin{equation}
\begin{split}
&\N{\rho_1-t\rho_2}_1\\
&\geq \frac{(M'-t)(1-m')+(t-m')(M'-1)}{M'-m'}\\
&= \frac{m'+M'-2m'M'+t\left(m'+M'-2)\right)}{M'-m'}\;,
\end{split}
\end{equation}
for all $t\in[m',M']$,
or equivalently, 
\be\label{eq:app-alb}
\N{\rho_1-t\rho_2}_1\geq \N{\sigma_1-t\sigma_2}_1\;,\qquad\forall\;t\geq 0\;,
\ee
where
\begin{align}\label{eq:sigmas}
&\sigma_1\equiv \frac{1}{M'-m'}
\begin{pmatrix}
	M'(1-m') & 0\\
	0 & m'(M'-1)
\end{pmatrix}\nonumber\\
&\sigma_2\equiv \frac{1}{M'-m'}
\begin{pmatrix}
	1-m' & 0\\
	0 & M'-1
\end{pmatrix}
\end{align}
are two diagonal $2\times 2$ density matrices with the property that $\sup(\sigma_1/\sigma_2)=M'\equiv\sup(\rho_1'/\rho_2')$ and $\inf(\sigma_1/\sigma_2)=m'\equiv\inf(\rho_1'/\rho_2')$. 

Finally, the statement of Theorem~\ref{theo:meas-and-prep} is obtained noticing that condition~(\ref{eq:app-alb}) is equivalent, due to Theorem~\ref{theo:1}, to saying that the Lorenz  curve of $\rho_1$ relative to $\rho_2$ is never below that of $\sigma_1$ relative to $\sigma_2$. However, since the latter is the quantum Lorenz curve of two classical probability distributions, it is just made of two segments joining the points $(0,0)$ with $(1,1)$, passing through either $\left(\frac{1-m'}{M'-m'},\frac{M'(1-m')}{M'-m'}\right)$ or $\left(\frac{M'-1}{M'-m'},\frac{m'(M'-1)}{M'-m'}\right)$, whichever determines the steepest curve.
\end{proof}

\subsection{Probabilistic transformations}

By mixing $\rho_1'$ with a sufficient fraction of $\rho_2'$, while keeping $\rho_2'$ unchanged, it is always possible to decrease the gap between $m'$ and $M'$, until the conditions of Theorem~\ref{theo:meas-and-prep} are met. In this way, with sufficient mixing, any output pair can be obtained, but the noise due to mixing cannot be undone afterwards.

One way to overcome this problem is to relax the assumptions made on the channel, in particular, the condition of trace-preservation. We hence consider \textit{probabilistic} channels of the following form:
\begin{equation}\label{eq:prob-test-prop}
\mE(\rho)\defeq\Tr[E\rho]\xi_1+\Tr[F\rho]\xi_2,
\end{equation}
where $E,F\ge0$, $E+F\le\openone$, and $\xi_1,\xi_2$ are two (normalized) density matrices. The above transformation constitutes a \textit{heralded} probabilistic transformation, in the sense that we know if the protocol succeeded or not, with success probability given by $\psucc=\Tr[(E+F)\rho]$.
The main result of this subsection is given by the following
\begin{theorem}\label{theo:probabilistic}
	Consider two pairs of density matrices $(\rho_1,\rho_2)$ and $(\rho_1',\rho_2')$ on $\mathbb{C}^n$ and $\mathbb{C}^m$, respectively. Then, a channel of the form~(\ref{eq:prob-test-prop}), such that
	\begin{equation}\label{eq:constraints}
	\mE(\rho_1)=p_1\rho_1',\quad\mE(\rho_2)=p_2\rho_2',
	\end{equation}
	exists if and only if
	\[
	\frac{m}{m'}\le\frac{p_1}{p_2}\le\frac{M}{M'},
	\]
	where $p_i=\Tr[(E+F)\rho_i]$, and $m$, $M$, $m'$, and $M'$ are as in Theorem~\ref{theo:meas-and-prep}.
\end{theorem}
When the protocol fails, we just prepare the state $\rho_2'$ independently of the input. In this way, we realized a channel that deterministically transforms $\rho_2$ into $\rho_2'$, but is also able to transform $\rho_1$ into $\rho_1'$, whenever the successful event is recorded.

In Ref.~\cite{reeb_hilberts_2011}, it is shown that a probabilistic transformation from $(\rho_1,\rho_2)$ to $(\rho'_1,\rho'_2)$ exists if and only if $M'/m'\le M/m$. The above theorem hence slightly extends that. For example, Theorem~\ref{theo:probabilistic} implies that the success probability $p_1$ be bounded as \[p_1\le M/M'\equiv  e^{-\Delta F_{\max}},\] where we define the \textit{max-free energy difference} as \[\Delta F_{\max}\defeq D_{\max}(\rho_1'\|\rho_2')-D_{\max}(\rho_1\|\rho_2).\]

In order to prove Theorem~\ref{theo:probabilistic}, let us consider a probabilistic test-and-prepare channel of the form:
\be\label{mpc2}
\Phi(\rho)=\Tr[E\rho]\sigma_1+\Tr\left[F\rho\right]\sigma_2
\ee
where $\sigma_1,\sigma_2$ are normalized density matrices and 
$E,F\geq 0$ with $E+F\leq \openone$. If $\mE(\rho_j)=p_j\rho_{j}'$ for $j=1,2$ with $0< p_1,p_2\leq 1$, then
\begin{align*}
p_1\rho_1'=e_1\sigma_1+f_1\sigma_2\\
p_2\rho_2'=e_2\sigma_1+f_2\sigma_2
\end{align*}
where $e_j\equiv\Tr[E\rho_j]$ and $f_j\equiv\Tr[F\rho_j]$ for $j=1,2$. Assuming $e_1f_2\ne e_2f_1$ the above equations are equivalent to
\begin{align*}
\sigma_1&=\frac{1}{e_1f_2-e_2f_1}\left[f_2p_1\rho_1'-f_1p_2\rho_2'\right]\\
\sigma_2&=\frac{1}{e_1f_2-e_2f_1}\left[-e_2p_1\rho_1'+e_1p_2\rho_2'\right]
\end{align*} 
Note that $\sigma_1$ and $\sigma_2$ have trace $1$ since $\rho_1'$ and $\rho_2'$ have trace 1.
Without loss of generality, we can assume $e_1/e_2>f_1/f_2$. With this choice, $\sigma_1$ and $\sigma_2$ are positive semi-definite if and only if
\be
\rho_1'-q^{-1}\frac{f_1}{f_2}\rho_2'\geq 0\quad\text{and}\quad
q^{-1}\frac{e_1}{e_2}\rho_2'-\rho_1'\geq 0\;,
\ee
where $q\defeq p_1/p_2$.
Again, denote by $m'=\inf(\rho_1'/\rho_2')$ and by $M'=\sup(\rho_1'/\rho_2')$. We therefore get that $\sigma_1$ and $\sigma_2$ are positive semi-definite if and only if
\be
\frac{f_1}{f_2}\leq qm' \quad\text{and}\quad
\frac{e_1}{e_2}\geq qM'\;.
\ee
Recalling the definitions of $e_1,e_2,f_1,f_2$, the above inequalities are equivalent to
\begin{align}
& \Tr\left[F\left(qm'\rho_2-\rho_1\right)\right]\geq 0\nonumber\\
& \Tr\left[E\left(\rho_1-qM'\rho_2\right)\right]\geq 0\;.
\end{align}
We therefore arrive at the following lemma:
\begin{lemma}
	There exists a CP map $\Phi$ of the form~\eqref{mpc2} such that $\Phi(\rho_1)=p_1\rho_1'$ and $\Phi(\rho_2)=p_2\rho_2'$, if and only if
	\be
	\Tr\left[F\left(qm'\rho_2-\rho_1\right)\right]\geq 0
	\ee
	and
	\begin{equation}
	\Tr\left[E\left(\rho_1-qM'\rho_2\right)\right]\geq 0\;.
	\end{equation}
\end{lemma}
Our goal is to maximize $p_1$ under these constraints along with the constraint $\tr\left[(E+F)(q\rho_2-\rho_1)\right]=0$ that defines $q$. Therefore, the maximum value of $p_1$, with a fixed value of $q\in\mbb{R}_{+}$, is given by
\be
P_{\max}(q)=\max_{\substack{\tr\left[(E+F)(q\rho_2-\rho_1)\right]=0\\
		\Tr\left[E\left(\rho_1-qM'\rho_2\right)\right]\geq 0\\
		\Tr\left[F\left(qm'\rho_2-\rho_1\right)\right]\geq 0\\
		E,F\geq 0\;,\;E+F\leq \openone_n}}\Tr\left[(E+F)\rho_1\right]
\ee
This is an optimization problem that can be solved efficiently and algorithmically using SDP.
Moreover, in the lemma below we show that if $q$ is not in the right interval then $P_{\max}(q)=0$.
\begin{lemma}
	$P_{\max}(q)>0$ implies that 
	$$
	\frac{m}{m'}\leq q\leq \frac{M}{M'}\;.
	$$ 
\end{lemma}
\begin{proof}
	Suppose $q<\frac{m}{m'}$. Then, $qm'<m$ so that $qm'\rho_2-\rho_1\leq 0$. Moreover,
	\begin{align*}
	&\Tr\left[F\left(qm'\rho_2-\rho_1\right)\right]\\
	&=
	\Tr\left[F\left(m\rho_2-\rho_1\right)\right]-(m-m'q)\tr[F\rho_2]\\
	& <0\;,
	\end{align*}
	unless $\tr[F\rho_1]=\tr[F\rho_2]=0$. We therefore must have $\tr[F\rho_1]=\tr[F\rho_2]=0$. This later condition
	gives
	\begin{align*}
	&\tr\left[(E+F)(q\rho_2-\rho_1)\right]=0\\
	&\iff \tr\left[E(q\rho_2-\rho_1)\right]=0\nonumber\\
	&\iff \tr\left[E\rho_1]=q\tr[E\rho_2\right]\;.
	\end{align*}
	But this last equality gives
	\be
	\begin{split}
	&\Tr\left[E\left(\rho_1-qM'\rho_2\right)\right]\\
	&=(1-M')\tr[E\rho_1]\\
	&\leq 0\;,
	\end{split}
	\ee
	since $M'>1$. We therefore must have $\tr[E\rho_1]=\tr[E\rho_2]=0$. Together with $\tr[F\rho_1]=\tr[F\rho_2]=0$, it gives $P_{\max}(q)=0$. Following similar lines we get $P_{\max}(q)=0$ for $q>M/M'$.
\end{proof}

As a consequence of the above discussion, we obtain the following corollary, which is consistent with a result in Ref.~\cite{reeb_hilberts_2011}, but slightly more general:

\begin{corollary}
	$$
	\max_{q\in\mbb{R}_{+}}P_{\max}(q)>0\;\;\iff\;\;\mathfrak{h}(\rho_1,\rho_2)\geq\mathfrak{h}(\rho_1',\rho_2')\;. 
	$$
\end{corollary}


\section{Discussion.}

In the present work we introduced quantum relative Lorenz curves and Hilbert $\alpha$-divergences, studied their properties, and applied them to the problem of characterizing necessary and sufficient conditions for the existence of a suitable transformation from an initial pair of states $(\rho_1,\rho_2)$ to a final one $(\rho_1',\rho_2')$. In particular, a strong equivalence has been proved in the case of coherent energy transitions with Gibbs-preserving maps, a paradigm that has immediate applications in quantum thermodynamics and the resource theory of athermality. Finally, we also considered the cases of test-and-prepare channels and probabilistic transformations, giving necessary and sufficient conditions for both.

\medskip\textit{Acknowledgments.}---The authors are grateful to Mark Girard for his help with the Figures~\ref{fig:region-class} and~\ref{fig:region-quantum}. F.B. acknowledges financial support from the JSPS KAKENHI, No. 26247016. G.G. acknowledges financial support from NSERC.

\appendix


\end{document}